\documentclass[british,english,hyphens]{article}
\usepackage[T1]{fontenc}
\usepackage[utf8]{inputenc}
\usepackage{babel}
\usepackage{array}
\usepackage{refstyle}
\usepackage{float}
\usepackage{booktabs}
\usepackage{units}
\usepackage{amsmath}
\usepackage{amsthm}
\usepackage{graphicx}
\usepackage{wasysym}
\usepackage[unicode=true]
 {hyperref}

\makeatletter

\AtBeginDocument{\providecommand\thmref[1]{\ref{thm:#1}}}
\AtBeginDocument{\providecommand\subsecref[1]{\ref{subsec:#1}}}
\providecommand{\tabularnewline}{\\}
\RS@ifundefined{subsecref}
  {\newref{subsec}{name = \RSsectxt}}
  {}
\RS@ifundefined{thmref}
  {\def\RSthmtxt{theorem~}\newref{thm}{name = \RSthmtxt}}
  {}
\RS@ifundefined{lemref}
  {\def\RSlemtxt{lemma~}\newref{lem}{name = \RSlemtxt}}
  {}

\numberwithin{equation}{section}
\numberwithin{figure}{section}
\theoremstyle{plain}
\newtheorem{thm}{\protect\theoremname}
  \theoremstyle{definition}
  \newtheorem{defn}[thm]{\protect\definitionname}
  \theoremstyle{remark}
  \newtheorem{rem}[thm]{\protect\remarkname}
  \theoremstyle{plain}
  \newtheorem{prop}[thm]{\protect\propositionname}
 \theoremstyle{definition}
 \newtheorem*{defn*}{\protect\definitionname}

\usepackage{hyperref}
\hypersetup{colorlinks = false, allcolors=., citebordercolor = [rgb]{0,0,0}, linkbordercolor = [rgb]{0,0,0}, urlbordercolor = [rgb]{0,0,0}}
\usepackage{pifont}
\usepackage[T1]{fontenc}
\usepackage{microtype}
\usepackage{lipsum}
\usepackage{blindtext,titlefoot}
\usepackage{ragged2e}
\usepackage{float}

\usepackage{pgfplots}
\pgfplotsset{width=7cm,compat=1.10}

\newtheorem*{assumption*}{\assumptionnumber}
\providecommand{\assumptionnumber}{}


\makeatother

  \addto\captionsbritish{\renewcommand{\definitionname}{Definition}}
  \addto\captionsbritish{\renewcommand{\propositionname}{Proposition}}
  \addto\captionsbritish{\renewcommand{\remarkname}{Remark}}
  \addto\captionsbritish{\renewcommand{\theoremname}{Theorem}}
  \addto\captionsenglish{\renewcommand{\definitionname}{Definition}}
  \addto\captionsenglish{\renewcommand{\propositionname}{Proposition}}
  \addto\captionsenglish{\renewcommand{\remarkname}{Remark}}
  \addto\captionsenglish{\renewcommand{\theoremname}{Theorem}}
  \providecommand{\definitionname}{Definition}
  \providecommand{\propositionname}{Proposition}
  \providecommand{\remarkname}{Remark}
\providecommand{\theoremname}{Theorem}

\begin{document}

\title{{\Large{}Pravuil: Global Consensus for a United World}}

\author{David Cerezo Sánchez\textsuperscript{}\\
{\small{}david@calctopia.com}}
\maketitle
\begin{abstract}
Pravuil \footnote{\label{fn:In-the-Book}In the Book of the Secrets of Enoch\cite{bookOfSecretsOfEnoch},
an archangel ``swifter in wisdom than the other archangels'', scribe
and recordkeeper.} is a robust, secure, and scalable consensus protocol for a permissionless
blockchain suitable for deployment in an adversarial environment such
as the Internet. Pravuil circumvents previous shortcomings of other
blockchains:

- Bitcoin's limited adoption problem: as transaction demand grows,
payment confirmation times grow much lower than other PoW blockchains

- higher transaction security at a lower cost

- more decentralisation than other permissionless blockchains

- impossibility of full decentralisation and the blockchain scalability
trilemma: decentralisation, scalability, and security can be achieved
simultaneously

- Sybil-resistance for free implementing the social optimum

- Pravuil goes beyond the economic limits of Bitcoin or other PoW/PoS
blockchains, leading to a more valuable and stable crypto-currency\\

\textbf{Keywords}: consensus, permissionless, permissioned, scalability,
zero-knowledge, mutual attestation, zk-PoI\\
\end{abstract}

\section{Introduction}

A third generation of blockchains has been developed featuring the
latest advances in cryptography and sharding to reach maximum performance
and security in Internet settings: they usually make use of advances
in BFT-like consensus protocols \cite{consensusSOK,shardingSOK} and
collective signatures\cite{cryptographySOK} to obtain 1000s of transactions
per second.

In this work, we introduce Pravuil \ref{fn:In-the-Book}, a robust,
secure, and scalable consensus protocol for real-world deployments
on open, permissionless environments that, unlike other proposals,
remains robust to high adversarial power and adaptation while considering
rational participants and providing strong consistency (i.e., no forks,
forward-security, and instant transactions). Our protocol is also
the first to integrate real-world identity on layer 1 as required
by current financial regulations, obtaining Sybil-resistance for free:
a very useful property considering the electrical waste produced by
Bitcoin, its Achilles' heel that this blockchain circumvents for the
first time by obviating to pay the Price of Crypto-Anarchy \cite{zkpoi}.

To achieve the desired goals, we introduce a new consensus protocol
in which we \foreignlanguage{british}{prioritise} robustness against
attackers and censorship-resistance. We then incorporate zero-knowledge
Proof-of-Identity \cite{zkpoi} while maintaining an open, permissionless
node membership mechanism enabling high levels of decentralisation.
Finally, we will show a working system of the proposed design in an
open-sourced Testnet at \href{https://github.com/Calctopia-OpenSource}{https://github.com/Calctopia-OpenSource}.

\subsection{Contributions}

In summary, we make the following contributions:
\begin{itemize}
\item we propose a consensus protocol that remains robust, secure, and scalable
among rational participants in an Internet setting
\item we prove liveness, safety, and censorship-resistance of our new consensus
protocol
\item we discuss the underlying rationale of our design and prove all the
advantages that it provides over previous blockchain designs
\item we provide an open-source implementation running on a Testnet
\end{itemize}

\section{Related Literature}

Previous blockchain designs \cite{consensusSOK,cryptographySOK,researchcryptocurrenciesSOK,shardingSOK}
deal with the different trade-offs of the scalability trilemma (security
vs. scalability vs. decentralisation) and they don't usually concern
with the economic consequences of their design (e.g., the Price of
Crypto-Anarchy) or the legal consequences of the lack of real-world
identity as required by recent legislation (FATF's Travel Rule).

Previous designs of ByzCoin/OmniLedger/MOTOR (\cite{byzcoin,omniledger,motor})
proposed Proof-of-Work(PoW) as a Sybil-resistance mechanism: although
their consensus protocol is more advanced and performant than Bitcoin,
they would still pay for the Price of Crypto-Anarchy \cite{zkpoi}.
And although other blockchains (e.g., \cite{concordium}) provide
methods to anonymise real-world identities, they fail to incorporate
these privacy techniques on their consensus protocol as they keep
on using Proof-of-Stake as a Sybil-resistance mechanism, thus they
still pay the Price of Crypto-Anarchy \cite{zkpoi}, suffer from Bitcoin's
limited adoption problem \cite{bitcoinFatalFlaw} and exist within
the same economic limits \cite{economicLimitBitcoin}.
\begin{center}
\begin{tabular}{|c|c|c|c|}
\hline 
 & \textbf{\small{}Bitcoin} & \textbf{\small{}ByzCoin/MOTOR} & \textbf{\small{}Pravuil}\tabularnewline
\hline 
\hline 
Secure & $\checked$ & $\checked$ & $\checked$\tabularnewline
\hline 
Decentralised & $\checked$ & $\checked$ & $\checked$\tabularnewline
\hline 
Scalability & $\XBox$ & $\checked$ & $\checked$\tabularnewline
\hline 
Real-world Identity & $\XBox$ & $\XBox$ & $\checked$\tabularnewline
\hline 
Free Sybil-resistance & $\XBox$ & $\XBox$ & $\checked$\tabularnewline
\hline 
Lawfulness & $\XBox$ & $\XBox$ & $\checked$\tabularnewline
\hline 
Unlimited adoption & $\XBox$ & $\XBox$ & $\checked$\tabularnewline
\hline 
No economic limitations & $\XBox$ & $\XBox$ & $\checked$\tabularnewline
\hline 
\end{tabular}
\par\end{center}

\section{Background and Model}

\subsection{Prior Work}

Pravuil builds over ByzCoin\cite{byzcoin}, OmniLedger\cite{omniledger},
and MOTOR \cite{motor}: in the next section \ref{sec:Detailed-Design},
we extend these protocols to address issues that prevent their deployment
in an adversarial environment such as the Internet.

\subsection{Assumptions}

In this work, we assume the following model and definitions:
\begin{defn}
\textbf{(Strongly-consistent broadcast \cite{parsimonious}).} A protocol
for strong consistent broadcast satisfies the following conditions
except with negligible probability:
\end{defn}

\begin{itemize}
\item \textbf{Termination}: If a correct party \textit{strongly-consistent
broadcasts} $m$ with tag $ID$, then all correct parties eventually
\textit{strongly-consistent deliver} $m$ with tag $ID$.
\item \textbf{Agreement}: If two correct parties $P_{i}$ and $P_{j}$ \textit{strongly-consistent
deliver} $m$ and $m'$ with tag $ID$, respectively, then $m=m'$.
\item \textbf{Integrity}: Every correct party \textit{strongly-consistent
delivers} at most one payload $m$ with tag $ID$. Moreover, if the
sender $P_{s}$is correct, then $m$ was previously \textit{strongly-consistent
broadcast} by $P_{s}$ with tag $ID$.
\item \textbf{Transferability}: After a correct party has \textit{strongly-consistent
delivered} $m$ with tag $ID$, it can generate a string $M_{ID}$such
that any correct party that has not \textit{strongly-consistent delivered}
message with tag $ID$ is able to \textit{strongly-consistent deliver}
some message immediately upon processing $M_{ID}$.
\item \textbf{Strong unforgeability}: For any $ID$, it is computationally
infeasible to generate a value $M$ that is accepted as valid by the
validation algorithm for completing $ID$ unless $n-2t$ correct parties
have initialised instance $ID$ and actively participated in the protocol.
\end{itemize}
\begin{defn}
\textbf{(Partial synchronous model \cite{partialSynchrony,partialSynchronyConsensus}).}
In a partially synchronous network, there is a known bound $\Delta$
and an unknown Global Stabilisation Time (GST), such that after GST,
all transmissions between honest nodes arrive within time $\Delta$.
\end{defn}

\begin{defn}
\textbf{(n=3f+1 \cite{impossibilityByzantineAgreeement}).} The proportion
of malicious nodes that an adversary controls accounts for no more
than $\nicefrac{1}{3}$ of the whole shard. The rest of the nodes
are rational, that is, maximisers of their transaction rewards.
\end{defn}

\begin{defn}
\textbf{(Round-adaptive adversary \cite{hybridConsensus}). }We assume
a mildly-adaptive, computationally bounded adversary that chooses
which nodes to corrupt at the end of every consensus round and has
control over them at the end of the next round.
\end{defn}

\begin{defn}
\textbf{(Strong Consistency \cite{byzcoin}).} The generation of each
block is deterministic and instant, with the following features:
\end{defn}

\begin{itemize}
\item There is no fork in a blockchain. By running a distributed consensus
algorithm, state machine replication is achieved.
\item Transactions are confirmed almost instantly. Whenever a transaction
is written into a block, the transaction is regarded as valid.
\item Transactions are tamper-proof (\textit{forward security}). Whenever
a transaction is written to a blockchain, the transaction and block
cannot be tampered with and the block will remain on the chain at
all times.
\end{itemize}
\begin{defn}
\textbf{(BLS \cite{blssignatures} and BDN \cite{bdnmultisignatures}
signatures). }Boneh-Lynn-Sacham and Boneh-Drijvers-Neven signatures
are assumed secure.
\end{defn}

\begin{defn}
\textbf{(Global PKI \cite{icaoPKI}).} Our blockchain design assumes
a global PKI, not directly for consensus purposes, but as a node-admission
and Sybil-resistance mechanism\cite{zkpoi}.
\end{defn}

\begin{defn}
\textbf{(Permissionless network \cite{whatispermisionless}).} In
a permissionless network:
\end{defn}

\begin{itemize}
\item Anyone can join a node without requiring permission from any party.
\item Any node can join or leave at any time.
\item The number of participating nodes varies at any time and is unpredictable.
\end{itemize}

\section{\label{sec:Detailed-Design}Detailed Design}

Pravuil builds over ByzCoin\cite{byzcoin}, OmniLedger\cite{omniledger},
and MOTOR \cite{motor}: ByzCoin\cite{byzcoin} envisions a Bitcoin\cite{bitcoin}
protocol that uses strongly consistent consensus, scaling with multi-cast
trees and aggregate Schnorr signatures. OmniLedger\cite{omniledger}
adds sharding over ByzCoin\cite{byzcoin}, and MOTOR\cite{motor}
strengthens the robustness of ByzCoin\cite{byzcoin} for an open,
adversarial network such as the Internet.

Pravuil improves over previous works by using another source of randomness,
drand\cite{drand}, and by incorporating zero-knowledge Proof-of-Identity\cite{zkpoi}
as a Sybil-resistance mechanism into the first layer of the consensus
protocol.

\subsection{Goals}

To sum up, Pravuil has the following goals:
\begin{itemize}
\item \textbf{Robustness}: the consensus round can only be disrupted by
controlling the leader node.
\item \textbf{Scalability}: the protocol performs well among hundreds of
nodes ($n=600$).
\item \textbf{Fairness}: the malicious leader can only be elected with a
probability equal to the percentage of malicious nodes in the system
(i.e., the adversary cannot always control the leader).
\end{itemize}
We detail the extensions over a previous BFT protocol such as ByzCoin/MOTOR
in order to obtain an improved blockchain-consensus algorithm.

\subsection{Rotating Leader\label{subsec:Rotating-Leader}}

View-change protocols assume a predetermined schedule of leaders,
making them susceptible to adversaries that compromise the next $f$
leaders.

To prevent this attack, our blockchain uses drand\cite{drand}: an
efficient randomness beacon daemon that \foreignlanguage{british}{utilises}
\foreignlanguage{british}{bilinear} pairing-based cryptography, t-of-n
distributed key generation, and threshold BLS\cite{blssignatures}
signatures to generate publicly-verifiable, unbiasable, unpredictable,
highly-available, distributed randomness at fixed time intervals.
As described in its online specification\cite{drandSpecification},
drand uses the BLS12-381 curve, the Feldman\cite{feldmanVSS} Verifiable
Secret Sharing protocol and the Joint Feldman protocol\cite{dkgDiscreteLog}
for DKG generation; using threshold BLS signatures as a source of
randomness is proven secure\cite{secureRandomBeacons} according to
its security model\cite{drandSecurityModel}.
\begin{rem}
In this work, we inherit all the previous security theorems from ByzCoin\cite{byzcoin},
OmniLedger\cite{omniledger}, and MOTOR \cite{motor}.
\end{rem}

\begin{thm}
(Robustness / Liveness). \label{thm:adversary-cannot-predict-leader-election}The
adversary cannot predict nor bias the leader election.
\end{thm}

\begin{proof}
The unpredictability property follows from the unforgeability of the
BLS\cite{blssignatures} signing algorithm, and the unbiasability
property follows from the deterministic nature of the BLS\cite{blssignatures}
signing algorithm. The leader of view $v$ is determined by the outcomes
of drand's public service, and all the nodes can publicly-verify its
election when needed. Thus, the adversary cannot predict nor bias
the leader election, preventing the adversary from breaking liveness.
\end{proof}
\begin{thm}
(Safety / Censorship-resistance). \label{thm:(Safety-/-Censorship-resistance)}
A round-adaptive adversary cannot always control the consensus decision.
\end{thm}

\begin{proof}
As the leader election is unpredictable (\thmref{adversary-cannot-predict-leader-election}),
the adversary can only hope that one of its randomly compromised nodes
gets chosen. Given that 
\[
\frac{1}{3^{d}}
\]
is the probability that the adversary controls $d$ consecutive leaders,
the adversary cannot control the leader forever since
\[
\underset{d\rightarrow\infty}{\lim}\frac{1}{3^{d}}=0
\]
thus the adversary always controls the consensus decision.
\end{proof}

\subsection{Zero-Knowledge Proof-of-Identity\label{subsec:Zero-Knowledge-Proof-of-Identity}}

In a previous work, we introduced zero-knowledge Proof-of-Identity\cite{zkpoi}
for biometric passports \cite{icaoPKI} and electronic identity cards
to permissionless blockchains in order to remove the inefficiencies
of Sybil-resistant mechanisms such as Proof-of-Work \cite{bitcoin}
and Proof-of-Stake \cite{ppcoin}. Additionally, attacks \cite{feasibilitySybilAttacks,abdelatif2021tractable}
on PoW sharded permissionless blockchains are prevented with zk-PoI:
an identity will be the same on all the shards, and the attacker can't
mine new identities for different shards as it's possible on PoW blockchains.

Although some could consider the latest zero-knowledge implementations
fast enough, their implementations are still too experimental for
production. For the first release, we will use the SGX implementation
based on mutual attestation, which works as follows (more details
on the original paper \cite{zkpoi}):

\begin{figure}[H]
\centering{}\includegraphics[scale=0.75]{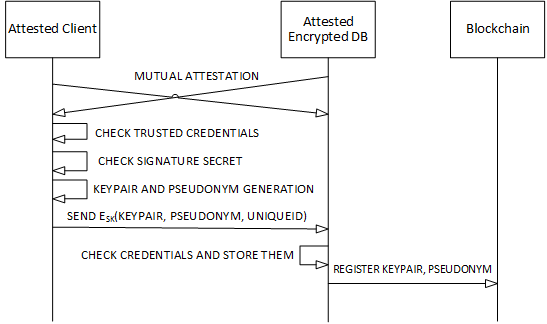}\caption{Simplified overview of mutual attestation protocol.}
\end{figure}

An approximate picture of the worldwide coverage follows:
\begin{center}
\begin{figure}[H]
\begin{centering}
\includegraphics[scale=0.3]{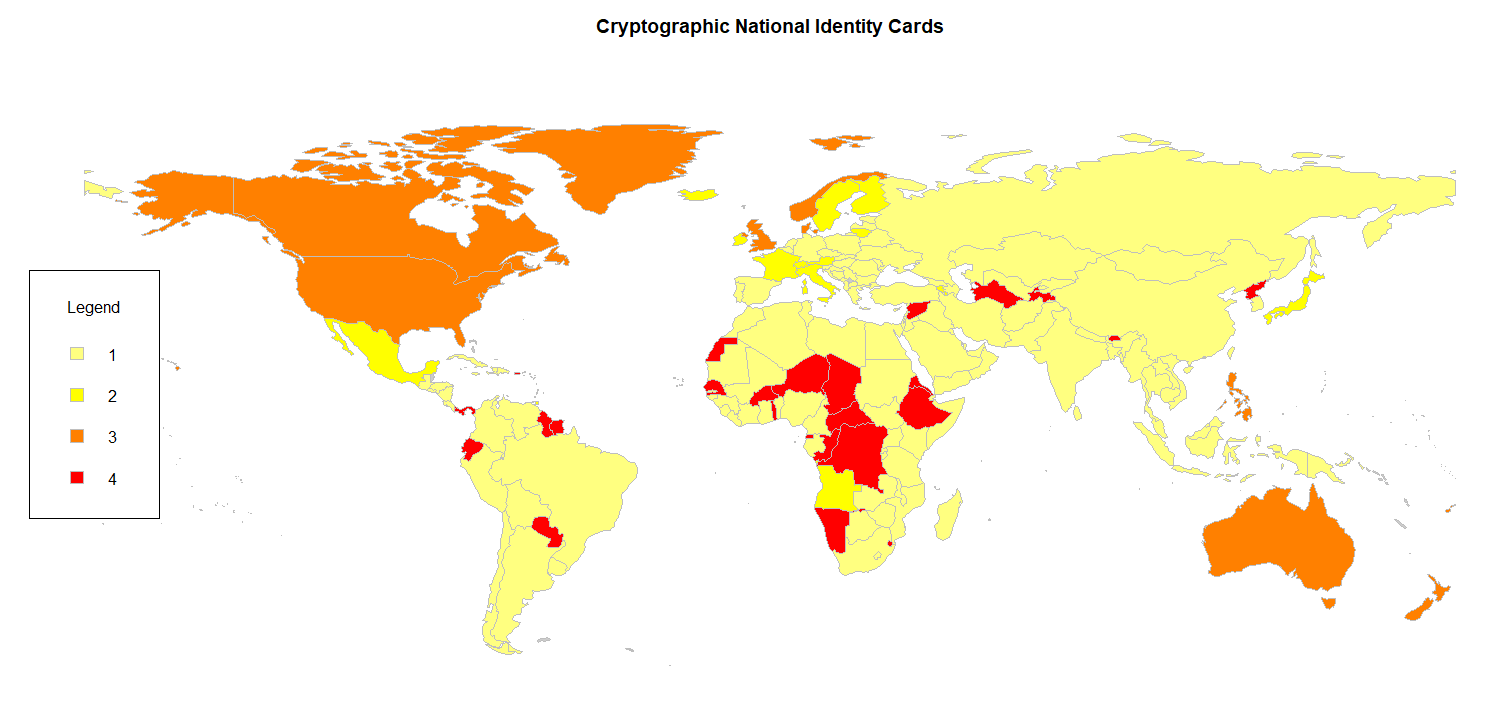}
\par\end{centering}
\caption{Legend: 1) National identity card is a mandatory smartcard; (2) National
identity card is a voluntary smartcard; (3) No national identity card,
but cryptographic identification is possible using an ePassport, driving
license and/or health card; (4) Non-digital identity card.}
\end{figure}
\par\end{center}

\section{Discussion}

In this section, we discuss the economic rationale underpinning the
unique features of this blockchain design that helps it to overcome
previous shortcomings and achieve an improved blockchain tailored
to real-world settings according to the experiences from the last
decade (e.g., Bitcoin\cite{bitcoin}).

\subsection{Overcoming Bitcoin's Limited Adoption Problem\label{subsec:Overcoming-Bitcoin-Limited}}

In a recent paper\cite{bitcoinFatalFlaw}, it is shown that a PoW
payments blockchain (i.e., Bitcoin) cannot simultaneously sustain
a large volume of transactions and a non-negligible market share:
\begin{prop}
(Adoption Problem \cite{bitcoinFatalFlaw}). Adoption decreases as
demand rises (i.e., the adoption rate of a network, $c^{*}$, decreases
in $N$). Moreover, the blockchain faces limited adoption,
\[
\underset{N\rightarrow\infty}{\lim}c^{*}=0.
\]
\end{prop}

Even allowing dynamic PoW supply (i.e., by relaxing PoW's artificial
supply constraint) achieves widespread adoption only at the expense
of decentralisation:
\begin{prop}
(Decentralisation implies Limited Adoption \cite{bitcoinFatalFlaw}).
PoW blockchains necessarily face either centralisation,
\[
\underset{N\rightarrow\infty}{\lim}\sup V\leq1,
\]
or limited adoption,
\[
\underset{N\rightarrow\infty}{\lim}c^{*}=0.
\]
\end{prop}

The previous propositions expose that the lack of widespread adoption
constitutes an intrinsic property of PoW payments blockchains: as
transaction demands grow, fees increase endogenously. Attracted by
this growth, more nodes join the validation process, expanding the
network size and thus protracting the consensus process and generating
increased payment confirmation times: only users insensitive to wait
times would transact in equilibrium, and limited adoption arises.
Moreover, this limitation cannot be overcome as it's rooted in physics
(i.e., network delay).

As pointed out by the previous proposition, centralised blockchains
overcome the limited adoption problem: for example, permissioned blockchains
that remain secure on an open, adversarial network such as the blockchain
proposed in this paper, enabling lower payment confirmation times
when omitting PoW's artificial supply constraint ,
\begin{prop}
(Lower Payment Confirmation Times \cite{bitcoinFatalFlaw}). For any
PoW protocol, there exists a permissioned blockchain that remains
secure on an open, adversarial network (i.e., Pravuil), which induces
(weakly) lower payment confirmation times.
\end{prop}

Additionally, omitting PoW's artificial supply constraint facilitates
timely service even for high transaction volumes:
\begin{prop}
(No Limited Adoption Problem \cite{bitcoinFatalFlaw}). In any permissioned
equilibrium, widespread adoption can be obtained,
\[
\underset{N\rightarrow\infty}{\lim}c_{P}^{*}=\min\left\{ \frac{R_{P}}{\Delta\left(V_{P}\right)},1\right\} .
\]
\end{prop}

\subsection{Obtaining Higher Transaction Security At A Lower Cost\label{subsec:Obtaining-Higher-Transaction}}

In another recent paper\cite{tradeoffsBlockchains}, it is shown that
permissioned blockchains have a higher level of transaction safety
than a permissionless blockchain, independent of the block reward
and the current exchange rate of the crypto-currency.

For a PoW permissionless blockchain, let $R$ be the block reward
in the corresponding crypto-currency, $x$ the associated exchange
rate to fiat currency, $w$ the block maturation rate (e.g., for Bitcoin,
$R=6,25;x=\$60.000;w=100$), $f$ be the probability of detecting
that blocks have been replaced, and $\beta_{pl}$ be the value above
which transactions are not safe,
\[
\beta_{pl}=fwRx.
\]

Note that 51\% attacks are becoming more common, specially for purely
financial reasons \cite{51attacks}. For a permissioned blockchain,
let $P_{i}$ be the punishment applied to each node $i$ if it participates
in an attack,$\tau\in\left[0,1\right]$ be the probability that nodes
that participated in an attack will be punished, and $\beta_{P}$
be the value above which transactions are not safe,
\[
\beta_{P}=f\tau\sum_{i\in B}P_{i},
\]
with $B$ being the set of $N$ nodes with the lowest $P_{i}$. Typical
punishments include confiscating all the funds deposited on the blockchain
and banning them from the blockchain, among others.
\begin{prop}
(\cite{tradeoffsBlockchains}). A permissioned blockchain that is
safe in an open, adversarial environment (i.e., Pravuil) has a higher
level of maximum value for transaction safety than a PoW permissionless
blockchain if
\[
\tau\sum_{i\in B}P_{i}>wRx.
\]
\end{prop}

Even with small values of $\tau$ will result in higher safety for
larger transactions than PoW permissionless blockchains:
\begin{prop}
(\cite{tradeoffsBlockchains}). For $\tau>0$ and high enough $P_{i}$'s,
a permissioned blockchain that is safe in an open, adversarial environment
(i.e., Pravuil) is more resilient than PoW permissionless blockchains
whenever
\[
\sum_{i\in B}P_{i}>\frac{wRx}{\tau}.
\]
\end{prop}

Ultimately, the cost of providing incentives to the validating nodes
not to participate in potential attacks (i.e., validating incentives
such as block rewards) will be lower for permissioned blockchains.
\begin{prop}
(\cite{tradeoffsBlockchains}). Suppose that $\beta_{pl}>0$ and $\beta_{p}>0$,
then at equilibrium the validator incentives in the permissioned blockchain
that is safe in an open, adversarial environment (i.e., Pravuil) are
lower than for the PoW permissionless.
\end{prop}

According to the model of this paper, in order to increase the transaction
safety, we only need to need increase:
\begin{itemize}
\item $\tau$, a probability that reflects user's trust in the system
\item $P_{i}$, a penalty that could also include legal action
\end{itemize}
In general, the mere existence of credible penalties $P_{i}$ with
positive probability $\tau$ is enough for the system to remain secure,
without needing to exert punishments in the case of rational attackers.
Additionally, note that these parameters are not economic parameters
of the system, unlike the parameters for PoW permissionless blockchains.

\subsection{An Empirical Approach to Blockchain Design}

Motivated by the abstract analysis from the previous sub-section \ref{subsec:Obtaining-Higher-Transaction},
we use the numerical comparisons between crypto-currencies from the
paper \cite{blockchainComparison} to compare permissionless and permissioned
blockchains in practice:

\begin{table}[H]
\centering{}%
\begin{tabular}{|c|c|c|c|c|c|c|c|}
\hline 
 & BTC & ETH & BCH & LTC & ADA & USDT & Mean\tabularnewline
\hline 
\hline 
Popularity & 1 & 2 & 4 & 7 & 8 & 9 & \tabularnewline
\hline 
Cost & 1.33 & 2 & 1.66 & 2.66 & 4.33 & 5 & 2.83\textbf{({*})}\tabularnewline
\hline 
Consistency & 1.33 & 2.33 & 1.33 & 2 & 3.66 & 1 & \tabularnewline
\hline 
Functionality & 2 & 5 & 2 & 2 & 4.33 & 2 & \tabularnewline
\hline 
Performance & 1.33 & 1.66 & 2 & 2.33 & 3 & 1 & 1.88\textbf{({*})}\tabularnewline
\hline 
Security & 4 & 4 & 4 & 4 & 4 & 3.33 & 3.88\tabularnewline
\hline 
Decentralisation & 5 & 3.33 & 4.33 & 3.66 & 3.33 & 1.33 & \tabularnewline
\hline 
Total & 14.99 & 18.32 & 15.32 & 16.7 & 22.65 & 13.66 & \tabularnewline
\hline 
 &  &  &  &  &  &  & \tabularnewline
\hline 
\textbf{\footnotesize{}Performance/Cost} & 0.28 & 0.41 & 0.46 & 0.7 & 1.79 & 1 & 0.77\textbf{({*})}\tabularnewline
\hline 
\textbf{(Perf{*}Sec)/Cost} & 1.13 & 1.66 & 1.84 & 2.79 & 7.18 & 3.33 & 2.99\textbf{({*})}\tabularnewline
\hline 
\textbf{Security/Cost} & 0.85 & 1 & 0.92 & 1.2 & 2.39 & 3.33 & 1.61\textbf{({*})}\tabularnewline
\hline 
\end{tabular}\caption{Permissionless blockchains. ({*}): p < 0.05 }
\end{table}

\begin{table}[H]
\begin{centering}
\begin{tabular}{|c|c|c|c|c|c|c|}
\hline 
 & XRP & EOS & XLM & TRX & MIOTA & Mean\tabularnewline
\hline 
\hline 
Popularity & 3 & 5 & 6 & 11 & 10 & \tabularnewline
\hline 
Cost & 4.66 & 5 & 4.66 & 5 & 5 & 4.84\textbf{({*})}\tabularnewline
\hline 
Consistency & 4.33 & 5 & 4 & 4 & 4.66 & \tabularnewline
\hline 
Functionality & 1.33 & 5 & 1.33 & 5 & 3.66 & \tabularnewline
\hline 
Performance & 4.33 & 4.66 & 4 & 4.66 & 5 & 4.53\textbf{({*})}\tabularnewline
\hline 
Security & 2.33 & 3.33 & 4 & 3.33 & 3.66 & 3.33\tabularnewline
\hline 
Decentralisation & 1 & 2.66 & 2.33 & 3.33 & 2.33 & \tabularnewline
\hline 
Total & 17.98 & 25.65 & 20.32 & 25.32 & 24.31 & \tabularnewline
\hline 
 &  &  &  &  &  & \tabularnewline
\hline 
\textbf{Performance/Cost} & 3.23 & 4.66 & 2.98 & 4.66 & 5 & 4.10\textbf{({*})}\tabularnewline
\hline 
\textbf{(Perf{*}Sec)/Cost} & 7.52 & 15.51 & 11.94 & 15.51 & 18.3 & 13.76\textbf{({*})}\tabularnewline
\hline 
\textbf{Security/Cost} & 1.73 & 3.33 & 2.98 & 3.33 & 3.66 & 3\textbf{({*})}\tabularnewline
\hline 
\end{tabular}
\par\end{centering}
\caption{Permissioned blockchains. \textbf{({*})}: p < 0.05 }
\end{table}
Using two-samples t-tests assuming unequal variances, we compare the
following means between permissionless and permissioned blockchains,
remarking that they are statistically significant:
\begin{itemize}
\item \textbf{Cost}: permissionless blockchains are costlier (2.83) than
permissioned blockchains (4.84). Please note that a higher cost score
means that the blockchain is considered to have better costs (i.e.,
lower costs), and the ranking obtained from this cost score must be
reversed to be useful in the next rankings.
\item \textbf{Performance}: permissionless blockchains are less performant
(1.88) than permissioned blockchains (4.53).
\item \textbf{Performance/Cost}: permissionless blockchains show worse performance
regarding cost (0.77) than permissioned blockchains (4.10).
\item \textbf{(Performance{*}Security)/Cost}: permissionless blockchains
show worse performance and security regarding cost (2.99) than permissioned
blockchains (13.76).
\item \textbf{Security/Cost}: permissionless blockchains show worse security
regarding cost (1.61) than permissioned blockchains (3).
\end{itemize}
It's clear from the empirical data that permissionless blockchains
are considered worse than permissioned blockchains when considering
cost, performance and security.

\subsection{Achieving More Decentralisation Than Other Permissionless Blockchains\label{subsec:Achieving-More-Decentralisation}}

In yet another recent publication \cite{permissionedDecentralization},
it is noticed that permissioned blockchains could achieve more decentralisation
than permissionless blockchains: real-world permissionless blockchains
are quite centralised \cite{gencer2018decentralization}, as there
aren't formal checks for the underlying centralisation. 

In order to obtain a more decentralised permissioned blockchain that
is safe in an open, adversarial network (i.e., Pravuil), the node
admission/gatekeeping function must be decentralised and opened: precisely,
this ideal state is achieved with our zero-knowledge Proof-of-Identity
\cite{zkpoi}, as previously explained in sub-section \ref{subsec:Zero-Knowledge-Proof-of-Identity}.

\begin{figure}[H]
\centering{}\includegraphics[scale=0.5]{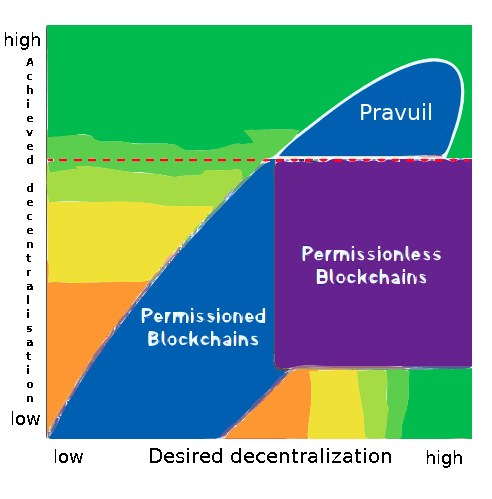}\caption{Comparing decentralisation (from \cite{permissionedDecentralization}).}
\end{figure}

\subsection{Overcoming The Scalability Trilemma}

The scalability trilemma postulates that a blockchain system can only
at most have two of the following three properties: decentralisation,
scalability, and security.
\begin{center}
\begin{figure}[H]
\begin{centering}
\includegraphics[scale=0.45]{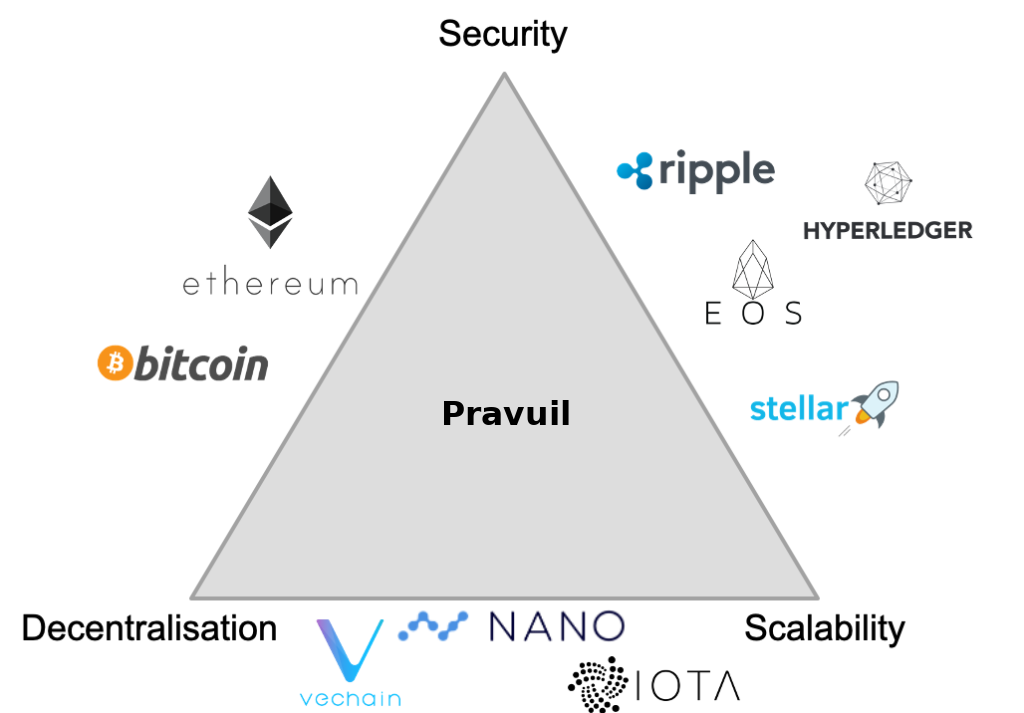}
\par\end{centering}
\caption{Pravuil overcomes the Scalability Trilemma.}
\end{figure}
\par\end{center}

In Pravuil, decentralisation, scalability, and security can be achieved
simultaneously:
\begin{itemize}
\item \textbf{Decentralisation}: as discussed in the previous sub-section
\ref{subsec:Achieving-More-Decentralisation}, Pravuil can be more
decentralised than other permissionless blockchains by using zero-knowledge
Proof-of-Identity, as previously explained in sub-section \ref{subsec:Zero-Knowledge-Proof-of-Identity}.
It also circumvents the impossibility of full decentralisation \cite{zkpoi}.
\item \textbf{Scalability}: Pravuil inherits the scalable \textit{Rotating-Subleader}
(RS) communication pattern from MOTOR \cite{motor}, specifically
created to avoid the communication bottleneck experienced by classic
BFT protocol when run over limited bandwidth networks.
\item \textbf{Security}: Pravuil is secure as previously proved in \thmref{adversary-cannot-predict-leader-election}
and \thmref{(Safety-/-Censorship-resistance)}.
\end{itemize}

\subsection{Obviating the Price Of Crypto-Anarchy of PoW/PoS Crypto-currencies}

In a previous paper \cite{zkpoi}, it was pointed out that the most
cost-efficient Sybil-resistant mechanism is the one provided by a
trusted national PKI infrastructure \cite{douceur2002the} and a centralised
social planner would prefer the use of National Identity Cards and/or
ePassports in order to minimise costs: instead, permissionless blockchains
are paying very high costs by using PoW/PoS as Sybil-resistant mechanisms.
The Price of Crypto-Anarchy compares the ratio between the worst Nash
equilibrium of the congestion game defined by PoW blockchains and
the optimal centralised solution, quantifying the costs of the selfish
behaviour of miners.
\begin{defn*}
\textbf{(\#26 from \cite{zkpoi}). }Let $NashCongestedEquil\subseteq S$
be the set of strategies given as the solution of the optimisation
problem of Theorem 25 from \cite{zkpoi}, then the \textit{Price of
Crypto-Anarchy} is given by the following ratio:
\[
\text{Price of Crypto-Anarchy}=\frac{\max_{s\in NashCongestedEquil}Cost\left(s\right)}{Cost\left(\text{zk-PoI}\right)}
\]
In practice, the real-world costs of Zero-Knowledge Proof-of-Identity
are almost zero as the identity infrastructure is subsidised by governments.
However, the situation for PoW/PoS blockchains is quite the opposite:
\end{defn*}
\begin{itemize}
\item PoW blockchains: in 2018, Bitcoin, Ethereum, Litecoin and Monero consumed
an average of 17, 7, 7 and 14 MJ to generate one US\$ \cite{bitcoinConsumption},
and in 2021 Bitcoin may be consuming as much energy as all data centers
globally \cite{bitcoinDataCenters,bitcoinBoom} at 100-130 TWh per
year. Holders of crypto-currency ultimately experience the Price of
Crypto-Anarchy as inflation from mining rewards, see next Table \ref{tab:Cost-of-decentralisation.}:
\end{itemize}
\begin{table}[H]
\begin{centering}
\begin{tabular}{c>{\centering}p{1cm}c>{\centering}p{1cm}c>{\centering}p{1.8cm}>{\centering}p{1.5cm}}
\toprule 
Name & Reward per block & Block time & Blocks per day & Price & \textbf{Yearly mining reward} & \textbf{Yearly Inflation}\tabularnewline
\midrule
\midrule 
BTC & 6,25 & 10 m & 144 & \$50000 & \textbf{\$18,061 B} & \textbf{4,12\%}\tabularnewline
\midrule 
ETH & 2 & 13,2 s & 6545 & \$3780 & \textbf{\$16,425 B} & \textbf{1,76\%}\tabularnewline
\midrule 
DOGE & 10,000 & 1 m & 1440 & \$0,49 & \textbf{\$2,575 B} & \textbf{4,06\%}\tabularnewline
\midrule 
LTC & 12,5 & 2,5 m & 576 & \$320 & \textbf{\$840 MM} & \textbf{3,94\%}\tabularnewline
\midrule 
BCH & 6,25 & 10 m & 144 & \$1275 & \textbf{\$418 MM} & \textbf{1,75\%}\tabularnewline
\midrule 
ZEC & 3,125 & 75 s & 1152 & \$301 & \textbf{\$395 MM} & \textbf{11,84\%}\tabularnewline
\midrule 
XMR & 1,02 & 2 m & 720 & \$407 & \textbf{\$109 MM} & \textbf{1,5\%}\tabularnewline
\bottomrule
\end{tabular}
\par\end{centering}
\caption{\label{tab:Cost-of-decentralisation.}Empirical Price of Crypto-Anarchy.}

\end{table}

\begin{itemize}
\item PoS blockchains: in theory, the costs are identical to the cost of
PoW schemes, except that instead of electrical resources and mining
chips, it takes the form of illiquid financial resources\cite{moreEconomicLimitsBlockchain}
and in practice, Proof-of-Stake is not strictly better than Proof-of-Work
as the distribution of the market shares between both technologies
has been shown to be indistinguishable (Appendix 3, \cite{dynamicsCrypto}).
\end{itemize}
Bitcoin miners have earned a total of \$26.75B as of April 2021: it's
not necessary to pay so much for Sybil resistance, instead, miners
could be paid for other tasks (e.g., transaction fees). As previously
discussed, obtaining Sybil-resistance for free is not only the key
to overcome Bitcoin's limited adoption problem (\subsecref{Overcoming-Bitcoin-Limited})
and to achieve more decentralisation than other permissionless blockchains
(\subsecref{Achieving-More-Decentralisation}), but also to go beyond
the economic limits of Bitcoin as discussed in the next \subsecref{Beyond-the-Economic}.

\subsection{Beyond the Economic Limits of Bitcoin\label{subsec:Beyond-the-Economic}}

In a paper about the economic limits of Bitcoin \cite{economicLimitBitcoin},
it is pointed out that Bitcoin is prohibitively expensive to run because
the recurring, ``flow'', payments to miners for running the blockchain
(particularly, the cost of PoW mining) must be large relative to the
one-off, ``stock'', benefits of attacking it. Let $V_{attack}$
be the expected payoff to the attacker, $P_{block}$ be the block
reward to the miner and $\alpha$ representing the duration of the
attack net of block rewards, then
\[
P_{block}>\frac{v_{attack}}{\alpha},
\]
placing serious economic constraints to the practicality and scalability
of the Bitcoin blockchain, a problem that seems intrinsic to any anonymous,
decentralised blockchain protocol. Consequently, the author poses
the open question of finding another approach to generating anonymous,
decentralised trust in a public ledger that is less economically expensive:
indeed, the technical solution hereby presented\ref{subsec:Zero-Knowledge-Proof-of-Identity}
that incorporates zero-knowledge Proof-of-Identity\cite{zkpoi} is
the technology that is both ``scarce and non-repurposable'', affordable
and not susceptible to sabotage attacks that could cause a collapse
in the economic value of the blockchain that the author of \cite{economicLimitBitcoin}
would seem meritorious to close said open question.

A more recent paper \cite{moreEconomicLimitsBlockchain} continues
the previous economic analysis \cite{economicLimitBitcoin}, extending
it to PoS and permissioned settings. For the permissionless PoS setting,
it finds that the costs are identical to the cost of PoW schemes,
except that instead of electrical resources and mining chips, it takes
the form of illiquid financial resources; however, zk-PoI\cite{zkpoi}
is free. For the permissioned case concerning this paper, if the block
reward is set exogenously, it finds that a permissioned blockchain
would have lower costs than permissionless PoW or PoS blockchains
in the economic model of \cite{economicLimitBitcoin}. 

\subsection{More Valuable and Stable Crypto-currencies}

A review of previous literature in economic research reveals the following
interesting facts regarding the intricate relationship between PoW
mining (i.e., hashrate, electricity and/or equipment costs) and crypto-currency
prices:
\begin{itemize}
\item There is a positive relationship between mining hashrate and price
\cite{determinantBitcoinPrices,cryptoValueFormation}: the causality
is primarily unidirectional going from the price to the hashrate \cite{hashratePrice},
although mining incidents and political shocks that affect mining
also negatively impact prices.
\item Bitcoin's security is sensitive (elastic) to mining rewards and costs,
although temporary mining cost and price shocks do not affect the
long-run blockchain security \cite{ciaian2021interdependencies}:
a 1\% permanent increase in the mining reward increases the underlying
blockchain security by 1.38\% to 1.85\% in the long-run; positive
shocks to electricity prices in China have a negative impact on the
hashrate in the short-run; a 1\% increase in the efficiency of mining
equipment increases the computing capacity between 0.23\% and 0.83\%
in the long-run; in the short-run mining competition intensity has
a statistically positive impact leading to expansion of mining capacity,
but in the long-run, the relationship is reversed.
\item High fixed mining rewards are the source of the instability to reach
an equilibrium between miners and users \cite{highfixedMiningRewards};
instead, mining rewards should be adjusted dynamically.
\item The production of crypto-currency by miners is jointly determined
with the price used by consumers \cite{equilibriumValuationBitcoinDN}:
the equilibrium price depends on both consumer preferences (i.e.,
price increases with the average value of censorship aversion, and
current and future size of the network) and the industrial organisation
of the mining market (i.e., price increases with the number of miners
and decreases with the marginal cost of mining). Price-security spirals
amplify demand and supply shocks: for example, a sudden demand shock
provoked by a government banning the crypto-currency in a country
would lead to price drops, itself leading to miners decreasing hashrate,
further decreasing prices and the feedback loop continuing until a
new equilibrium is reached in multiple rounds. In other words, Bitcoin's
security model embeds price volatility amplification.
\item In a PoW blockchain, it's impossible to simultaneously achieve all
the three following goals \cite{decentralizingMoney}: maximise crypto-currency
price, blockchain's security, and social welfare.
\end{itemize}
Similar results can be found for PoS blockchains because they are
substituting electricity and mining costs for illiquid and volatile
financial resources \cite{moreEconomicLimitsBlockchain}. In general,
the interdependencies can be described graphically as the following
cycles and spirals:
\begin{center}
\begin{figure}[H]
\begin{centering}
\includegraphics[scale=0.5]{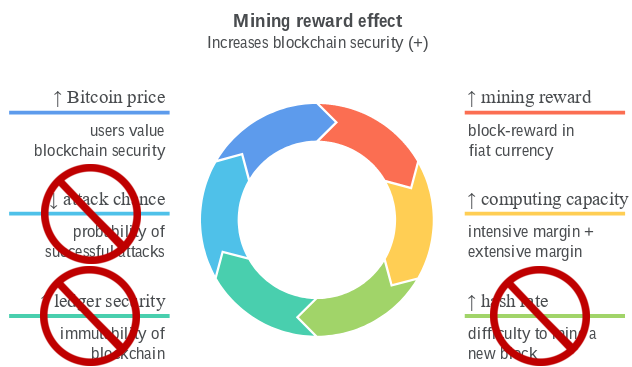}
\par\end{centering}
\caption{Interdependencies \cite{ciaian2021interdependencies}, with broken
negative feedback loops.}
\end{figure}
\par\end{center}

\begin{center}
\begin{figure}[H]
\begin{centering}
\includegraphics[scale=0.4]{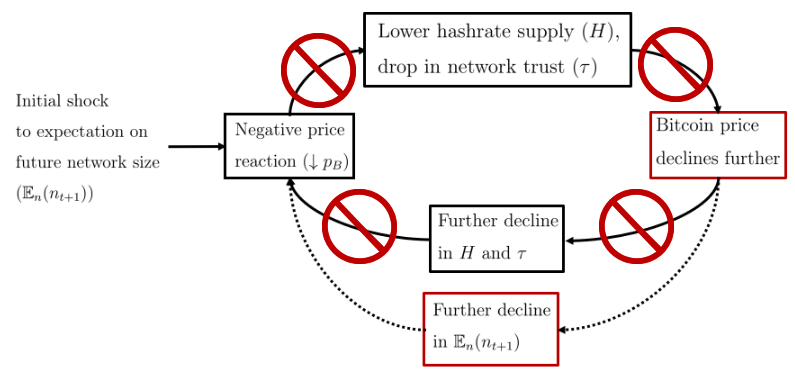}
\par\end{centering}
\caption{Spirals \cite{equilibriumValuationBitcoinDN}, with broken negative
feedback loops.}
\end{figure}
\par\end{center}

\begin{flushleft}
However, we break most of the previous interdependencies and spirals
with our strongly-consistent blockchain with free Sybil-resistance:
\par\end{flushleft}
\begin{itemize}
\item blockchain and transaction security are independent of blockchain
mining capacity, mining costs and rewards, and price: once a transaction
is instantly committed, it's committed forever.
\item there aren't price-security spirals for demand and supply shocks:
changes in prices do not lead to changes in security.
\item as blockchain's security is independent of price, it's possible to
maximise crypto-currency price and social welfare.
\end{itemize}
Ultimately, our blockchain design leads to more valuable and stable
crypto-currencies.

\section{Implementation}

Pravuil has a Testnet deployed with a working implementation consisting
of:
\begin{itemize}
\item a blockchain layer in Go and Java, invoking drand \cite{drand} as
described in this paper \ref{subsec:Rotating-Leader}.
\item zero-knowledge Proof-of-Identity \cite{zkpoi} in Python and C.
\item mobile apps for Android (Typescript, Java) and iOS (Typescript, Objective-C,
Swift).
\item secure smart contracts in Obliv-Java \cite{raziel}.
\end{itemize}
All the code will be open-sourced at \href{https://github.com/Calctopia-OpenSource}{https://github.com/Calctopia-OpenSource},
including future developments.

\section{Conclusion}

In this work, we presented Pravuil, an improvement over previous blockchains
that is suitable for real-world deployment in adversarial networks
such as the Internet. Pravuil achieves this feat by:
\begin{itemize}
\item unpredictably rotating leaders using drand \cite{drand} to defend
against adversaries and censorship attacks: drand is an Internet service
that generates publicly-verifiable, unbiasable, unpredictable, highly-available,
distributed randomness at fixed time intervals.
\item using for the first time zero-knowledge Proof-of-Identity \cite{zkpoi}
as a Sybil-resistance mechanism to overcome Bitcoin's limited adoption
problem\cite{bitcoinFatalFlaw} and to go beyond the economic limits
of Bitcoin\cite{economicLimitBitcoin}, delivering more decentralisation
than other permissionless blockchains \cite{permissionedDecentralization}.
\item based on the design of a blockchain layer that scales-out with strong
consistency prioritising robustness over scalability.
\end{itemize}
\bibliographystyle{alpha}
\bibliography{bib}

\end{document}